\documentclass[11pt,letterpaper]{article}

\usepackage{amsmath}
\usepackage{amsthm}
\usepackage{amssymb}
\usepackage{mathrsfs}
\usepackage{marvosym}

\usepackage{multirow}
\usepackage{amsfonts}
\usepackage[colorlinks]{hyperref}
\usepackage{geometry}
\usepackage{setspace}
\usepackage{textcomp}
\usepackage{enumerate}
\usepackage{authblk}
\usepackage[ruled,linesnumbered]{algorithm2e}

\usepackage{graphicx}
\DeclareGraphicsExtensions{.pdf,.jpg,.png}

\newtheorem{conj}{Conjecture}[section]
\newtheorem{thrm}{Theorem}[section]
\newtheorem{coro}{Corollary}[section]
\newtheorem{defn}{Defenition}[section]

\newtheorem*{remk*}{Remark} 
\newtheorem{lema}{Lemma}[section]

\begin{document}

\title{The PCP-like Theorem for Sub-linear Time Inapproximability}

\author[1]{Hengzhao Ma}
\author[2,\Letter]{Jianzhong Li}
\affil[1]{hz.ma@siat.ac.cn}
\affil[2]{lijzh@siat.ac.cn}
\affil[1,2]{Shenzhen Institute of Advanced Technology, Chinese Academy of Sciences}

\date{}
\maketitle              

\section*{Abstract}

In this paper we propose the PCP-like theorem for sub-linear time inapproximability. Abboud et al. have devised the distributed PCP framework for proving sub-quadratic time inapproximability. Here we try to go further in this direction. Staring from SETH, we first find a problem denoted as Ext-$k$-SAT, which can not be computed in linear time, then devise an efficient MA-like protocol for this problem. To use this protocol to prove the sub-linear time inapproximability of other problems, we devise a new kind of reduction denoted as Ext-reduction, and it is different from existing reduction techniques. We also define two new hardness class, the problems in which can be computed in linear-time, but can not be efficiently approximated in sub-linear time. Some problems are shown to be in the newly defined hardness class.

\noindent\\
\textbf{Keywords:} PCP Theorem, Sub-linear Time, Inapproximability, SETH.

\clearpage

\section{Introduction}\label{sec:intro}

One of the most important task of computation theory is to identify the hardness of problems, and thus distinguish the hard problems with the easy ones. Problems with similar hardness are grouped into complexity classes, and the complexity classes form a plentiful of hierarchy, such as EXPTIME, PSPACE, NP and P, all the way down to the hierarchy. In common sense the problems in $P$, which are the problems that can be solved by Turing machine in polynomial time, are thought to be practically tractable. The most famous and important open problem, which is to prove or falsity $P=NP$, reveals the common belief that polynomial time solvable means easy, otherwise hard. That is why the discovery of polynomial time algorithm for linear programming is such a great shock \cite{Khachiyan1980}, which was conjectured to be in NP in quite a period of time. 

However, we are now in a time of massive data. The problem we are facing may remain the same, but the size of data grows tremendously. Even the most basic problems like sorting are granted terabytes of data. Using elaborately designed parallel algorithms, sorting 100TB of data takes around 100 seconds, but in serial environment the same work may consume more than days \cite{SortBenchmark}. Another example is the deep learning tasks. Fed with millions of data records, the training phase of deep learning systems usually lasts days or months. In theoretical point of view, these problems are all solvable in polynomial time, but in some circumstances consuming days of time may not be the original intention of \textit{tractable}. What's more, what if more data need to be dealt with, like petabyes ($2^{40}$B),  exabytes ($2^{50}$B) and even more? In fact scientific data is growing near or out of this order of magnitude. For example, the Large Synoptic Survey Telescope (LSST) \cite{Ivezic2019} can generate 1.25 PB of data per day, and over 3 years the accumulated data would exceed 1EB. It is no doubt that how hard it is to find something valuable in such amount of data. 

Reflecting the practical examples into theory, the following problem arises: is it still appropriate to take polynomial time as the criterion of \textit{tractable}, under the massive data circumstances? Let us now consider the following simple calculation for linearly reading data from SSD hard drive. The maximum reading bandwidth for SSD's is around 8GB/s, so merely reading 1PB data needs 34.7 hours. If the data size reaches 1EB, then it takes more than 4 years to linearly scan them all. This example shows that even linear-time algorithms could be too slow when facing with massive data. Thus, \textit{sub-linear} time may be appropriate for the criterion of tractable in massive data circumstances.

Actually sub-linear time algorithm is not a fresh new research area. Perhaps the most familiar sub-linear time algorithm is binary search on sorted array, which takes $O(\log{n})$ time. The time complexity can be further reduced to $O(\log{\log{n}})$ if using interpolation search \cite{Li1987}. The two search strategies can find the exact result, but for most problems sub-linear time rules out exact results, since reading the entire input is impossible. Thus to achieve sub-linear time it is rational to consider approximate algorithms. 

Now consider the situation that a researcher wants to devise a sub-linear time approximate algorithm for one specific problem $P$. Before he even starts to work on designing the algorithm, one primary question he must answer is that, dose sub-linear time approximate algorithm exist for problem $P$? Or equivalently, is $P$ sub-linear time approximable? However, so far there is no theoretical tools for proving sub-linear time inapproximability. 
It is well known that polynomial time inapproximability is based on the PCP theorem \cite{Arora1998a}. Hence, as a natural comparison, one would ask, if there exists a PCP-like theorem for sub-linear time inapproximability?

We notice that Abboud et al. proposed the distributed PCP framework in \cite{Abboud2017}. It is based on the Strong Exponential Time Hypothesis, and most results proved by distributed PCP are sub-quadratic inapproximability. It is the first PCP-like theorem for proving inapproximability results in $P$. The existing works have shown that the distributed PCP can be used to prove $O(n^{k-\epsilon})$-time inapproximability for arbitrary integer $k\ge 2$, but it can not apply to $O(n^{1-\epsilon})$-time inapproximability as far as we know. In order to finish the last step between PCP theorems and sub-linear time inapproximability, we propose our PCP-like theorem in this paper. The main contributions of this paper are listed below.

\begin{enumerate}
	\item We propose the PCP-like theorem for proving sub-linear time inapproximability, which is denoted as Ext-PCP theorem. This completes the theoretical toolsets for proving hardness of approximation in P.
	\item Using the Ext-PCP theorem, several problems are proved to be sub-linear time inapproximable. The proof involves a new kind of reduction called Ext-reduction. Different with existing reduction techniques, the Ext-reduction is based on the probabilistic verification system of the Ext-PCP theorem.
	\item Summarizing the proved sub-linear time inapproximability results, we define two kinds of hardness classes, the problems in which can be computed in linear time, but can not be efficiently approximated in sub-linear time.
\end{enumerate}

The rest of this paper is organized as follows. We first review some related works in Section \ref{sec:rwork}. The Ext-PCP theorem is proposed in Section \ref{sec:sublinear-pcp}, which is the main contribution of this paper. 
As several applications of the sub-linear PCP theorem, in Section \ref{sec:app} we prove some new sub-linear time inapproximability results. Based on the obtained results, we propose a new reduction for proving sub-linear time inapproximability, denoted as Ext-reduction, and propose two hardness class which are Strictly Linear Time Problems and Parameterized Linear Time Problems. Finally Section \ref{sec:conc} concludes this paper.

\section{Related works}\label{sec:rwork}

\subsection{PCP theorem and hardness of approximation}
The PCP (Probabilistically Checkable Proofs) theorem \cite{Arora1998a,Arora1998} may be the most important result ever since the NP-complete theory is found \cite{Cook1971,Karp1972,Levin1973}. Formally, a PCP system is a probabilistic polynomial-time Turing machine $M$, which is given an input $x$ and an array of bits $\Pi$ called the proof string. The verifier is able to randomly access the proof string, i.e., given an address string $i$ the verifier can directly read the bit $\Pi[i]$. The verifier works as follows. It tosses random coins to acquire a random string $\tau$, reads a number of positions in $\Pi$, and decide to accept or reject. Define $M^{\Pi}(x,\tau)=1$ if M accepts $x$ using $\tau$ after examining $\Pi$, otherwise $M^{\Pi}(x,\tau)=0$. 

There are two parameters that determine the power of PCP systems, which are the length of the random string $r(n)$, and the number of positions in the proof string that allowed to be queried $q(n)$. We say a verifier is $(r(n),q(n))$-constrained if it uses $O(r(n))$ bits of random string, and queries at most $O(q(n))$ bits from the proof string. Then define $PCP(r(n),q(n))$ to be a language class as follows.

\begin{defn}
	A language $L\in PCP(r(n),q(n))$ iff there exists a $(r(n),q(n))$-constrained verifier $M$ which behaves as below for every input $x$.
	\begin{itemize}
		\item If $x\in L$, there exists a proof $\Pi$ that causes $M$ to accept for every random string, i.e., 
		$\mathop{Pr}\limits_{\tau}[M^{\Pi}(x,\tau)]=1.$
		\item If $x\notin L$, the for all proofs $\Pi$, 
		$\mathop{Pr}\limits_{\tau}[M^{\Pi}(x,\tau)]\le \frac{1}{2}.$
	\end{itemize}
\end{defn}

The following is the PCP theorem, due to Arora et al. \cite{Arora1998}.
\begin{thrm}[\cite{Arora1998}]\label{thrm:classic-pcp}
	$NP=PCP(\log{n},1).$
\end{thrm}

The impressive power of the PCP theorem is that it not only provides a new characterization of NP, but also acts as the basic tool for proving NP-hardness of approximation problems. The basic result is for $\rho$-gap-$k$-SAT, which is defined as follows.

\begin{defn}[$\rho$-gap-$k$-SAT]
	Given a $k$-CNF formula $\varphi$ with $n$ variables and $m$ clauses, the problem is to distinguish the following two situations:
	\begin{itemize}
		\item \textbf{(\textbf{Completeness})} $\varphi$ is satisfiable, i.e., there exists an assignment $\alpha$ such that $\varphi(\alpha)=1$.
		\item \textbf{(\textbf{Soundness})} $\varphi$ is not satisfiable but what's more, for every possible assignment $\alpha$ at most $\rho m$ number of clauses can be satisfied.
	\end{itemize}
\end{defn}

A reduction from every language $L\in NP$ to $\rho$-gap-$k$-SAT can be constructed with the aid of the PCP Theorem \ref{thrm:classic-pcp}, showing that there exists $\rho>0$ such that approximating $k$-SAT problem within approximation ratio $\rho$ is NP-hard \cite{Arora1998a}. A strong characteristic of the inapproximability results proved via PCP theorem is that, there is usually a \textit{gap} between the 'YES' instances and 'No' instances.

With the hardness of $\rho$-gap-$k$-SAT proved, a lot of other inapproximability results can be reached via polynomial time reductions, such as vertex cover \cite{Papadimitriou1988}, metric TSP \cite{Papadimitriou1993}, Steiner Tree \cite{Bern1989}, et. See \cite{Arora2003} for a good survey.

\subsection{Fine-grained complexity}
The fine-grained complexity focuses on proving running time lower bounds for problems in $P$. In the history of researchers trying to devise more efficient algorithms for every problem, there are many problems that further improvement on the time complexity is seemingly impossible. For NP-complete problems, using polynomial time reductions from 3SAT and assuming $P\ne NP$, it can be proved that no polynomial time algorithm exists for a specific problem, such as vertex cover, traveling salesman, et. But for problems known to have polynomial time algorithms, proving the lower bound for them needs a wholly different tool set. Basically we need the following primary components. A hypothesis stating a time lower bound on a specific well-studied problem, and a kind of reduction used to prove the lower bound for other problems. There are three wildly used hypotheses, which are SETH, 3SUM and APSP, and the reduction used here is called fine-grained reduction.

\begin{conj}[SETH]\label{conj:seth}
	For any $\epsilon>0$ there exists $k\ge 3$ such that $k$-SAT problem with $n$ variables can not be solved in $O(2^{(1-\epsilon)n})$ time. 
\end{conj}

\begin{conj}[3SUM]
	The 3SUM problem is given a set $S$ of $n$ integers from $\{-n^c,\cdots,n^c\}$ for some constant $c$, and to determine whether there exists $x,y,z\in S$ such that $x+y+z=0$. The hypothesis is that for $c=4$, 3SUM problem can not be solved by any randomized algorithm in $O(n^{2-\epsilon})$ time for any $\epsilon>0$.
\end{conj}

\begin{conj}[APSP]
	APSP stands for All Pair Shortest Path problem. Given a graph $G=(V,E)$ with $n$ nodes and integral edge weights $w:E\rightarrow \{-M,\cdots,M\}$, where $M=poly(n)$, APSP is to compute for every $u,v\in V$, the smallest distance $d(u,v)$ in G, where the distance is defined as the sum of weights along the shortest path from $u$ to $v$. The conjecture says that no randomized algorithm can solve APSP in $O(n^{3-\epsilon})$ time for any $\epsilon>0$ on a graph without negative cycles.
\end{conj}

\begin{defn}[Fine-grained reduction]
	Let $A$ and $B$ be two computational problems, and assume $a(n)$ and $b(n)$ be their conjectured running time lower bounds, respectively. We say $A$ can be $(a,b)$-reduced to $B$, denoted as $A\le_{a,b} B$, if for every $\epsilon >0$, there exists $\delta>0$, and an algorithm $R$ for $A$ that runs in $a(n)^{1-\delta}$ time, making $q$ calls to an oracle algorithm for $B$ with query lengths $n_1,\cdots,n_q$, where 

	$$\sum\limits_{i=1}^{q}{(b(n_i))^{1-\epsilon}} \le  (a(n))^{1-\delta}.$$
	
	If $A\le_{a,b}B$ and $B\le_{b,a}A$, then we say that $A$ and $B$ are fine-grained equivalent, $A\equiv_{a,b}B$.
\end{defn}

Due to space limitation we only briefly introduce some lower bound results based on SETH. The most important result may be on $k$-Orthogonal Vectors ($k$-OV), which is defined below.

\begin{defn}[$k$-OV]
	For $k\ge 2$, $d=\omega(\log{n})$, given $k$ sets $A_1,\cdots,A_k\subseteq \{0,1\}^d$ with $|A_i|=n$ for $1\le i\le k$, determine whether there exist $a_1\in A_1,\cdots,a_k\in A_k$ such that $a_1\cdot\ldots\cdot a_k=0$ where $a_1\cdot\ldots\cdot a_k:=\sum\limits_{i=1}^{d}\prod\limits_{j=1}^{k}{a_j[i]}$.
\end{defn}

There is a fine-grained reduction from SAT which proves that $k$-OV can not be solved in $O(n^{k-\epsilon})$ time even by randomized algorithms, assuming SETH \cite{Williams2018}. 

There are a lot of problems that are fine-grained equivalent to $k$-OV, such as Batch Subset Query problem\cite{Agrawal2010,Goel2010,Melnik2003,Ramasamy2000}. Some other problems can be reduced from $k$-OV, such as graph diameter \cite{Roditty2013}, closest pair in Hamming space \cite{Alman2015}, Longest Common Subsequence \cite{Abboud2015,Bringmann2015}, and Dominating Set \cite{Karthik2019}, et. See \cite{Williams2018} for a full survey.

\subsection{Sub-linear time algorithms}\label{subsec:sub-linear-algs}
Initially the design of sub-linear time algorithms focused on Property Testing, which is another kind of definition for approximation. The property testing problem is based on the definition of $\epsilon$-far from a specific property $P$. 
\begin{defn}
	An input $x$, represented as a function $x:\mathcal{D}\rightarrow \mathcal{R}$, is $\epsilon$-close to satisfying property $P$, if there exists some $y$ satisfying $P$ and $x$ differs with $y$ on at most $\epsilon|\mathcal{D}|$ places in their representation. Otherwise $x$ is said to be $\epsilon$-far from satisfying $P$.
\end{defn}

\begin{defn}\label{defn:property-test}
	Let $P$ be a property, and let $c,s$ be constants in $[\frac{1}{2},1]$. For input $x$ of size $n=|\mathcal{D}|$ and parameter $\epsilon$, a property tester for $P$ must satisfy:
	\begin{itemize}
		\item (\textbf{Completeness}) If $x$ satisfies property $P$, the tester must accept $x$ with at least $c$ probability.
		\item (\textbf{Soundness}) If $x$ is $\epsilon$-far from satisfying $P$, the tester must reject $x$ with at least $s$ probability.
	\end{itemize}
\end{defn}

There are easier problems such as testing the monotonicity of an array \cite{Ergun1998}, and more complex problems such as testing the diameter of graphs \cite{Parnas1999}. See \cite{Rubinfeld2011} for more related works.

Recently many sub-linear time algorithms are found in other areas, such as computational geometry, graph problems, and algebraic problems. Here we list some examples.
For computational geometry problems, there is an $O(\sqrt{n})$ time algorithm for deciding whether two polygons intersect \cite{Chazelle2005}. For graph problems, there is an $O(\sqrt{n}/\epsilon)$ algorithm for approximating the average degree within a factor of $2+\epsilon$ \cite{Feige2006}, and there is an algorithm for approximating the weight of the minimum spanning tree within a factor of $\epsilon$ which runs in $\tilde{O}(D\cdot W/\epsilon^2)$ time, where $D$ is the maximum degree and $W$ is the maximum weight \cite{Chazelle2005a}. For algebraic problems, testing whether a function $f:\mathcal{D}\rightarrow \mathcal{R}$ is a homomorphism \cite{Blum1993}  has been an important problem, and it is also used in the construction of PCP systems \cite{Arora1998}. See \cite{Rubinfeld2011} for more related works.

\subsection{Distributed PCP}

The distributed PCP \cite{Abboud2017} is the work mostly close to this work, and it is also the most important inspiration of this work. Unlike the classic PCP system, in the distributed PCP system there are four, not two, parties, which are named Alice, Bob, Merlin and Veronica (the verifier). Alice and Bob each holds half of the input without knowing each other's content, and Merlin acts as an all-knowing advice provider. The model originates from the Arthur-Merlin communication model \cite{Aaronson2009}, and is inspired by the truth that most problems related to SETH are given two parts with equal positions as input. The distributed PCP theorem in \cite{Abboud2017} is recited below.

\begin{thrm}[Distributed PCP]\label{thrm:distributed-pcp}
	Let $\varphi$ be a boolean CNF formula with $n$ variables and $m=O(n)$ clauses. There is a non-interactive protocol where:
	\begin{itemize}
		\item Alice, given the CNF $\varphi$, partial assignment $\alpha\in\{0,1\}^{n/2}$, and advice $\mu\in\{0,1\}^{o(n)}$, outputs a string $a^{\alpha,\mu}\in \{0,1\}^{2^{o(n)}}$.
		\item Bob, given $\varphi$ and partial assignment $\beta\in \{0,1\}^{n/2}$, outputs a string $b^{\beta}\in \{0,1\}^{2^{o(n)}}$.
		\item The verifier, given input $\varphi$, tosses $o(n)$ coins, non-adaptively reads $o(n)$ bits from $b^{\beta}$, and adaptively reads one bit from $a^{\alpha,\mu}$; finally, the verifier returns Accept or Reject.
	\end{itemize}
	If the combined assignment $(\alpha,\beta)$ satisfies $\varphi$, there exists advice $\mu^*$ such that the verifier always accepts. Otherwise, i.e., if $\varphi$ is unsatisfiable, for every $\mu$ the verifier rejects with probability at least $1-1/2^{n^{1-o(1)}}$.
\end{thrm}

\subsection{Pure sub-linear vs. Pseudo sub-linear}
In \cite{Gao2020} the authors distinguished the two concepts which are pure sub-linear tractable and pseudo sub-linear tractable. Informally, \textit{pure} sub-linear tractable means the problem can be directly solved by a sub-linear time algorithm, where \textit{pseudo} sub-linear allows a polynomial time preprocessing. For example, finding an element in an unsorted array is pseudo sub-linear tractable, since we can pay $O(n\log{n})$ time on sorting the array as processing, and then find the element using binary search in $O(\log{n})$ time. In this sense, all the sub-linear time algorithms listed in Section \ref{subsec:sub-linear-algs} are pure sub-linear time algorithms. For pseudo sub-linear time algorithms, the authors in \cite{Fan2013} defined the $\sqcap$-tractability class, which is the class of boolean query problems that can be solved in parallel polylogarithmic time after a polynomial time preprocessing. As far as we know, the two works \cite{Fan2013,Gao2020} are the only ones that distinguished the concept of pure and pseudo sub-linear time.

Transferring the concept onto approximation algorithms, it is obvious that pure sub-linear time approximability is more strict than pseudo sub-linear time approximability. On the other hand, pseudo sub-linear time inapproximability is more strict than pure sub-linear time inapproximability. In fact, there are several existing results based on SETH and distributed PCP, showing that some problem is pseudo sub-linear time inapproximable. For example, in \cite{Abboud2017} the authors proved a result for Max Inner Product, which says that there exists $\rho<1$ such that no algorithm can preprocess the input in polynomial time then answer a query in $O(n^{1-\epsilon})$ time within a approximation factor of $\rho$. To this extent, the distributed PCP framework seems to be strong enough for proving pure sub-linear time inapproximability.
However, the distributed PCP framework can fail in certain cases, where the sub-linear PCP theorem proposed in this paper suffices. These results will be discussed in Section \ref{sec:app}. In summary, it can be said that the two kinds of PCP system for approximation in P both have their strength and limitation.

\section{Preliminaries}

\subsection{Sparsification Lemma}

The sparsification lemma is important in the following discussions. Roughly speaking, the lemma states that any $k$-CNF can be expressed by another $k$-CNF whose number of clauses is linear in the number of variables.

\begin{lema}[\cite{Impagliazzo2001}]\label{lema:sparsification}
	For all $\epsilon>0$, a $k$-CNF $F$ can be expressed as the disjunction of at most $2^{\epsilon n}$ $k$-CNF $F_i$ such that $F_i$ contains each variable in at most $c(k,\epsilon)$ clauses for some function $c$, Moreover, this reduction takes at most $poly(n)2^{\epsilon n}$ time.
\end{lema}

Together with SETH (Conjecture \ref{conj:seth}) we have the following corollary.

\begin{coro}\label{coro:spar-seth}
	If there exists $\epsilon>0$ such that for all $k\ge 3$, $k$-SAT on $n$ variables and $c_{k,\epsilon}n$  clauses can be solved in $O(2^{(1-\epsilon)n})$ time, then SETH is false.
\end{coro}

\subsection{Polynomial encoding}
One of the most important building blocks of the PCP theorem is the polynomial codes. Here we do not intend to describe the details, but only cite the following lemma, which is central to analyze the soundness of PCP theorems.

\begin{lema}[Schwartz-Zippel Lemma]\label{lema:schwartz-zipple}
	Let $F_1(x_1,x_2,\cdots,x_n),F_2(x_1,x_2,\cdots,x_n)$ be two multivariate polynomials of maximum degree $d$ defined over a field $\mathbb{F}$. Fix any finite set $S\subset \mathbb{F}$, and let $r_1,r_2,\cdots,r_n$ be chosen independently and uniformly at random in $S$. Then 
	
	$$\Pr\left[F_1(r_1,r_2,\cdots,r_n)=F_2(r_1,r_2,\cdots,r_n) \mid F_1\ne F_2\right]\le \frac{d}{|S|} $$
\end{lema}

\section{The PCP theorem for sub-linear inapproximability}\label{sec:sublinear-pcp}

\subsection{Existential $k$-Satisfaction Problem}
We first propose the Existential $k$-Satisfaction (Ext-$k$-SAT) problem, which is the corner stone of the discussion in this section.

\begin{defn}\label{defn:ESAT-prpblem}
	Given a set $A=\{\alpha_1,\cdots, \alpha_n\}$ of assignments, where $|\alpha_i|=d$ and $d=\Omega(\log{n})$, $d=o(n)$, and a $k$-CNF $\varphi$ with $d$ variables and $m=O(d)$ clauses, decide whether there exists an assignment $\alpha^*\in S$ such that $\varphi$ is satisfied by $\alpha^*$. 
\end{defn}

The following Theorem \ref{thrm:ESAT-no-linear} gives some insight of the hardness of Ext-$k$-SAT.
	
\begin{thrm}\label{thrm:ESAT-no-linear}
	Assuming SETH, for $\forall \epsilon >0$, there exists $k\ge 3$ such that Ext-$k$-SAT problem can not be solved in $O(n^{1-\epsilon})$ time.
\end{thrm}

\begin{proof}
	Recalling SETH, for $\forall \epsilon >0$, there exists $k\ge 3$ such that $k$-SAT can not be solved in $O(2^{(1-\epsilon)n})$ time. Thus, for arbitrary given $\epsilon >0$, let $k$ be the value stated in SETH, and consider the Ext-$k$-SAT problem. If Ext-$k$-SAT can be solved in $O(n^{1-\epsilon})$ time, then the SAT problem with $n$ variables can be solved in $O(2^{(1-\epsilon)n})$ time by the following algorithm. Let $S$ be the set of all possible assignments, then $|S|=2^n$. Now just invoke the $O(n^{1-\epsilon})$ time algorithm for Ext-$k$-SAT, then SAT problem is solved in $O(|S|^{1-\epsilon})=O(2^{(1-\epsilon)n})$ time, contradicting with SETH.
\end{proof}

The Ext-$k$-SAT problem can be regarded as the hardness core of SAT if using linear enumeration.
Next we proceed to design an MA-protocol for Ext-$k$-SAT problem, which will be given in Section \ref{subsec:ext-k-sat-protocol}. We first consider the set containment problem.

\subsection{The protocol for Set Containment problem}

\begin{defn}
	Given a set $S$ of elements over a universe $\mathbb{U}$ and a query element $e\in \mathbb{U}$, decide whether $e\in S$.
\end{defn}

\begin{thrm}\label{thrm:set-ctnmt-protocol}
	For $T\le |\mathbb{U}|$, there exists a protocol for Set Containment problem, where Alice holds the set $S$ and Bob holds the query element $e$. The protocol works as follows:
	\begin{itemize}
		\item Bob sends Alice $O(|\mathbb{U}|\log{|\mathbb{U}|}/T+T)$ bits;
		\item Alice toss $O(\log{|\mathbb{U}|})$ coins, and returns Accept or Reject.
	\end{itemize}

	If the input set $S$ contains the query element $e$, then Alice always returns Accept; otherwise ($e\notin S$), Alice will reject with a probability at least $\frac{1}{2}$. The probability is over the output of random coin tossing.
\end{thrm}

\begin{proof}
	Let $q$ be a prime number such that $4|\mathbb{U}|\le q\le 8|\mathbb{U}|$, and let $\mathbb{F}_q$ be the prime field with size $q$. For a efficient communication protocol we consider the following arithmetization. Assume without loss of generality that $T$ divides $|\mathbb{U}|$, and represent the set that Alice holds as $T$ functions $\psi_{s,t}:[\frac{n}{T}]\rightarrow \{0,1\}$ as follows: $\psi_{s,t}(x)\triangleq 1$ if and only if the element corresponding to $(x,t)$ is in Alice's set. The query element can be regarded as a set with a single element, and thus can be represented as $T$ functions $\psi_{e,t}$ in the same way. 
	
	We first consider the situation that $T=1$, where the set and the element are both represented by one single function, which are $\psi_s(x)$ and $\psi_e(x)$ respectively. Consider the product $\Psi(x)=\psi_s(x)\cdot \psi_e(x)$. If $S$ contains $e$, then it can be verified that $\Psi(x)=1$ only when $x=e$, otherwise $\Psi(x)=0$. In this situation $\Psi(x)$  exactly equals to $\psi_e(x)$. If $S$ does not contain $e$, then $\Psi(x)\equiv 0$. 
	
	Now consider arbitrary $T$ but assume without generality that $T$ divides $|\mathbb{U}|$. The protocol works as follows.
	
	\begin{enumerate}
		\item Bob sends Alice the following message: $0_1,0_2,\cdots, \psi_{e,t},\cdots, 0_T$.
		\item Alice draws $x \in \mathbb{F}_q$ uniformly at random.
		\item Alice reads through message, and when encounters $\psi_{e,t}$ she conducts the following examination:
		\begin{equation}\label{eqtn:sublinear-pcp-one}
			\psi_{s,t}(x)\cdot\psi_{e,t}(x)=\psi_{e,t}(x)
		\end{equation}
		and accepts if and only if the examination succeeds.
	\end{enumerate}

	We now explain about Bob's message. The message that Bob sends to Alice represents the query element $e$. Since $\psi_{e,t}(x)$ equals 1 only when $(x,t)$ corresponds to $e$, the parts irrelevant to $e$ can be represented by one-bit $0$. Now consider $\psi_{e,t}(x)$ that corresponds to $e$. The degree of it is at most $\frac{|\mathbb{U}|}{T}-1$, and thus it can be uniquely identified by $\frac{|\mathbb{U}|}{T}-1$ coefficients in $\mathbb{F}_q$. Since each coefficient requires $\log_2{|\mathbb{F}_q|}=\log_2{|\mathbb{U}|}+O(1)$ bits, the length of the message that Merlin sends is $O(|\mathbb{U}|\log{|\mathbb{U}|}/T+T)$.
	
	Finally we analyze the completeness and soundness.
	
	\textbf{Completeness} If Alice's set contains the element, then Alice always accepts.
	
	\textbf{Soundness} If the element is not contained in Alice's set, then $\psi_{s,t}(x)\cdot\psi_{e,t}(x)$ differs with $\psi_{e,t}(x)$. By the Schwartz-Zippel Lemma (Lemma \ref{lema:schwartz-zipple}), since the degree of the two polynomials are no larger than $2\frac{|\mathbb{U}|}{T}\le q/2$, the two polynomials must differ on at least half of $\mathbb{F}_q$. And thus (\ref{eqtn:sublinear-pcp-one}) is false with probability at least $\frac{1}{2}$.
\end{proof}

In the above Theorem \ref{thrm:set-ctnmt-protocol} we use $|\mathbb{U}|$ as the complexity parameter, and now we discuss the range of $|\mathbb{U}|$ represented by $|S|$. First we must have $|\mathbb{U}|=\Omega(S)$, otherwise there must be repeated elements in $S$ by the Pigeonhole Principle. Then we add a restriction that $|\mathbb{U}|=o(2^{|S|})$, since if $|\mathbb{U}|=\Theta(2^{|S|})$ then the length of each element in $S$ is comparable to the number of elements in $S$, which is unusual. Now denoting $d=\log{|\mathbb{U}|}$ which is the length of the binary representation of each element in $\mathbb{U}$, and $n=|S|$, we have $d=\Omega(\log{n})$ and $d=o(n)$. This is also the parameter range given in Definition \ref{defn:ESAT-prpblem}.

Back to Theorem \ref{thrm:set-ctnmt-protocol}, we show that the $|\mathbb{U}|$ term can be reduced to $O(n)$. In the above proof we try to encode the entire universe $\mathbb{U}$, but it is not necessary. It suffices to encode only the elements in $S$. Thus we can choose a smaller field $\mathbb{F}_q$ of size $O(n)$, and then the complexity of the protocol can be reduced.

\begin{coro}\label{coro:set-ctnmt-protocol}
	For $T\le n$ where $n=|S|$, there exists a protocol for Set Containment problem, where Alice holds the set $S$ and Bob holds the query element $e$. The protocol works as follows:
	\begin{itemize}
		\item Bob sends Alice $O(n\log{n}/T+T)$ bits;
		\item Alice tosses $O(\log{n})$ coins, and returns Accept or Reject.
	\end{itemize}
	
	If the input set $S$ contains the query element $e$, then Alice always returns Accept; otherwise (if $e\notin S$), Alice will reject with a probability at least $\frac{1}{2}$. The probability is over the output of random coin tossing.
\end{coro}

In the above protocol Alice holds the set and Bob holds the element. The situation can be reversed, i.e., Alice holds the element and Bob holds the set.

\begin{thrm}\label{thrm:set-ctnmt-reverse}
	For $T\le n$ where $n=|S|$, there exists a protocol for Set Containment problem, where Alice holds the query element $e$ and Bob holds the set $S$. The protocol works as follows:
	\begin{itemize}
		\item Bob sends Alice $O(n\log{n}/T)$ bits;
		\item Alice tosses $O(\log{n})$ coins, and returns Accept or Reject.
	\end{itemize}
	
	If the input set $S$ contains the query element $q$, then Alice always returns Accept; otherwise (if $e\notin S$), Alice will reject with a probability at least $\frac{1}{2}$. The probability is over the output of random coin tossing.
\end{thrm}

\begin{proof}
	Let the arithmetization of the set $S$ and the element $e$ be the same with what given in the proof of Theorem \ref{thrm:set-ctnmt-protocol}. Since $S$ is represented by $T$ functions $\psi_{s,t}:[\frac{n}{T}]\rightarrow \{0,1\}$, we let $\Psi_s=\sum\limits_{t\in [T]}\psi_{s,t}$. On the other hand let $\Psi_e=\sum\limits_{t\in [T]}\psi_{e,t}=0+0+\cdots+\psi_{e,t^*}+\cdots +0=\psi_{e,t^*}$, where $t^*$ is the only value that makes $\psi_{e,t}$ not zero. Now consider the product 
	$$\Phi(x)=\Psi_s(x)\cdot \Psi_e(x)=\Psi_s(x)\cdot \psi_{e,t^*}(x)=\sum\limits_{t\in [T]}\psi_{s,t}(x)\cdot\psi_{e,t^*}(x)$$
	
	By the discussion in the proof of Theorem \ref{thrm:set-ctnmt-protocol}, if $S$ contains $e$ then the product $\psi_{s,t}\cdot\psi_{e,t^*}$ equals $\psi_{e,t^*}$, otherwise the product equals 0. Thus it can be verified that, if $S$ contains $e$ then $\Phi(x)=\psi_{e,t^*}(x)=\Psi_e(x)$, otherwise $\Phi(x)\ne\Psi_e(x)$.
		
	According to the above analysis, we give the following protocol for the set containment problem in the reversed situation.
	
	\begin{enumerate}
		\item Bob sends Alice $\Psi_s$.
		\item Alice draws $x\in \mathbb{F}_q$ uniformly at random.
		\item Alice checks the equality $\Psi_s(x)\cdot\Psi_e(x)=\Psi_e(x)$, and accepts if and only if the equality holds.
	\end{enumerate}
	
	The completeness and soundness analysis is similar with Theorem \ref{thrm:set-ctnmt-protocol} and thus omitted.
\end{proof}

\subsection{The protocol for Ext-$k$-SAT}\label{subsec:ext-k-sat-protocol}

Now we consider the Ext-$k$-SAT problem. Now Alice holds the $k$-CNF $\varphi$, and Bob holds the set of assignments $S$. We will show that there exists an efficient MA-like protocol for Ext-$k$-SAT problem, and it is also the key to the sub-linear PCP theorem. 
\begin{thrm}\label{thrm:ESAT-protocol-base}
	There exists an MA-protocol for the Ext-$k$-SAT problem which works as follows.
	\begin{itemize}
		\item Merlin sends Alice   $O(n\log{n}/T + d\cdot polylog(d))$ bits.
		\item Bob sends Alice $O(n\log{n}/T)$ bits.
		\item Alice tosses $O(\log{n})$ coins, and returns $Accept$ or $Reject$.
	\end{itemize}
	
	If there exists an element $\alpha^*\in A$ such that $\alpha^*$ satisfies $\varphi$, then there exists a message from Merlin which causes Alice certainly returns $Accept$. Otherwise, i.e., if all elements in $A$ can not satisfy $\varphi$, then for any possible message from Merlin, Alice would return $Reject$ with at least $\frac{1}{2}$ probability. 
\end{thrm} 

\begin{proof}
	The message $\beta$ that Bob sends to Alice is the same with $\Psi_s$ described in Theorem \ref{thrm:set-ctnmt-reverse}, and the length of $\beta$ is $O(n\log{n}/T)$. 
	
	The message that Merlin sends to Alice consists of three parts $\mu_e, \mu_o, \mu_c$, which are used to test set containment, CNF satisfaction, and the consistency of the former two parts, respectively. 
	
	The first part $\mu_e$ is the same with $\Psi_e(x)$ described in Theorem \ref{thrm:set-ctnmt-reverse}, where $\Psi_e$ is supposed to represent $\alpha^*$ which is the satisfactory assignment. The length of $\Psi_e$ is $O(n\log{n}/T)$, as described in Theorem \ref{thrm:set-ctnmt-reverse}.
	
	As for the second part, recall the classical PCP theorem \cite{Arora1998}. If a CNF $\varphi$ with $d$ variables is satisfiable, then there exists an oracle $\Pi$ such that the verifier can generate a random string with $O(\log{d})$ bits, randomly read $O(1)$ bits from $\Pi$ and always return Accept. Otherwise, for any oracle $\Pi$, the verifier can generate a random string with $O(\log{d})$ bits, randomly read $O(1)$ bit from $\Pi$ and reject with a probability at least $\frac{1}{2}$. And according to \cite{Dinur2007}, the length of $\Pi$ is $O(d\cdot polylog(d))$. After all, the second part of Merlin's message $\mu_o$ can be the same $\Pi$ with the classical PCP theorem, and its length is $O(d\cdot polylog(d))$.
	
	The third part $\mu_c$ is for testing the consistency of $\mu_e$ and $\mu_o$, i.e., testing whether $\mu_e$ and $\mu_o$ are generated based on the same assignment. The most easy way is to set $\mu_c$ to be the same with some assignment, since it can fulfill the consistency test and does not exceed the $O(d\cdot polylog(n))$ length limitation. To test if $\mu_e$ is generated based on $\mu_c$, simply test if $\mu_e(\mu_c)=1$ holds, where $\mu_e$ is considered as the function $\Psi_e(x)$ and $\mu_c$ is considered as an assignment $\alpha$. The way to test $\mu_o$ is based on the proofs in \cite{Arora1998} and is omitted here.

	With the messages from Merlin and Bob obtained, Alice conducts the following verification. She uses $\mu_e$ and $\beta$ to conduct the set containment verification as described in Theorem \ref{thrm:set-ctnmt-reverse}, uses $\mu_o$ to conduct the $k$-CNF satisfiability verification as described in \cite{Arora1998}, and uses $\mu_c$ to verify that $\mu_e$ and $\mu_o$ correspond with the same assignment $\mu_c$.

	Now we analyze the completeness and soundness of the whole system.
	
	If there exists an element $\alpha^*\in A$ such that $\alpha^*$ satisfies $\varphi$, Merlin can send $\mu_e$ as the true $\Psi_{\alpha^*}$, $\mu_o$ as the true PCP oracle for $\alpha^*$, and $\mu_c$ as the same with $\alpha^*$ . Under these messages, Alice would return $Accept$ for sure. 
	
	If $\forall \alpha\in A$, $\varphi$ cannot be satisfied by $\alpha$, we discuss the following situations. If $\varphi$ is satisfiable by $\alpha^*$ but $\alpha^*\notin A$, and Merlin sends the messages as the true strings corresponding to $\alpha^*$, the satisfiability verification will succeed with probability 1, but the set containment verification will fail with probability at least $\frac{1}{2}$. If Merlin sends $\mu_o$ corresponding  $\alpha^*$ but sends $\mu_e$ corresponding to some $\alpha'\in A$, the consistency verification will fail with probability at least $\frac{1}{2}$. If $\varphi$ is not satisfiable, and Merlin sends the messages that correspond to an assignment $\alpha'\in A$, the set containment test will succeed with probability 1, but the satisfiability verification would fail with probability at least $\frac{1}{2}$. We omit the detailed discussion about the other cases, but after all in all cases, at least one of the three verifications would fail with probability at least $\frac{1}{2}$. Thus in conclusion, if $\forall \alpha\in A$ can not satisfy $\varphi$\, Alice will reject with probability at least $\frac{1}{2}$.
\end{proof}

\begin{coro}\label{coro:ESAT-protocol-final}
	For any given $\epsilon,\delta>0$, there exists an MA-protocol for the Ext-$k$-SAT problem which works as follows.
	\begin{itemize}
		\item Merlin sends Alice $O(n^{1-\epsilon}d)$ bits.
		\item Bob sends Alice $O(n^{1-\epsilon}\log{n})$ bits.
		\item Alice tosses $O(n^{1-\delta}\log{n})$ coins, and returns $Accept$ or $Reject$.
	\end{itemize}
	
	If there exists an element $\alpha^*\in A$ such that $\alpha^*$ satisfies $\varphi$, then there exists a message from Merlin which causes Alice certainly returns $Accept$. Otherwise, i.e., if all assignments in $A$ can not satisfy $\varphi$, then for any possible message from Merlin, Alice would return $Reject$ with at least $1-1/2^{n^{1-\delta}}$ probability. 
\end{coro}

\begin{proof}
	The protocol is obtained by repeatedly execute the one given in Theorem \ref{thrm:ESAT-protocol-base}. Setting $R=n^{1-\delta}$ and $T=n^\epsilon$, repeat the protocol for $R$ times to amplify the soundness. Note that Merlin and Bob send their message to Alice before the random coins are tossed, thus Merlin and Bob need only to send the message once. Substituting  $T=n^\epsilon$ and note that $d=\Omega(\log{n}), d=o(n)$, we obtain that Merlin's message is $O(n\log{n}/T + d\cdot polylog(d))=O(n^{1-\epsilon}d+d\cdot polylog(n))=O(n^{1-\epsilon}d)$, and Bob's message is $O(n\log{n}/T)=O(n^{1-\epsilon}\log{n})$. The total number of random coins is $O(\log{n}\cdot R)=O(n^{1-\delta}\log{n})$. Letting the reject probability of Alice in Theorem \ref{thrm:ESAT-protocol-base} be $p$, the probability of Alice rejects in the repeated protocol is $1-(1-p)^R$, and with $p\ge \frac{1}{2}$ it can be derived that $1-(1-p)^R\ge 1-1/2^{R}= 1-1/2^{n^{1-\delta}}$.
\end{proof}

\subsection{The Ext-PCP theorem}
Based on the Ext-$k$-SAT problem (Definition \ref{defn:ESAT-prpblem}) and the MA-protocol for it (Corollary \ref{coro:ESAT-protocol-final}), we introduce the following PCP-like theorem. It may be denoted as Ext-PCP, or Sublinear-PCP.

\begin{thrm}\label{thrm:sublinear-pcp}
	Let $A=\{\alpha_1,\cdots,\alpha_n \}$ be a set of $n$ assignments, where $|\alpha_i|=d$ and $d=\Omega(\log{n}),d=o(n)$. Let $\varphi$ be a boolean CNF formula with $d$ variables and $m=O(d)$ clauses. For any given $\epsilon,\delta>0$, there is a non-interactive protocol where:
	\begin{itemize}
		\item Bob, given the set $A$ of assignments, outputs a bit string $\beta\in \{0,1\}^{O(n^{1-\epsilon}\log{n})}$.
		\item Alice, given the CNF $\varphi$, and advice $\mu\in \{0,1\}^{O(n^{(1-\epsilon)d})}$, outputs a string $a^{\beta,\mu}\in \{0,1\}^{2^{O(n^{1-\delta}\log{n})}}$.
		\item The verifier, given the CNF $\varphi$, tosses $O(n^{1-\delta}\log{n})$ coins, non-adaptively reads one bit from $a^{\beta,\mu}$, and returns Accept or Reject.
	\end{itemize}
	If there exists an assignment $\alpha^*\in A$ that satisfies $\varphi$, then there exists advice $\mu^*$ such that the verifier always accepts. Otherwise, i.e., all $\alpha\in A$ can not satisfy $\varphi$, or $\varphi$ is not satisfiable, then for arbitrary advice $\mu$ the verifier rejects with probability at least $1-1/2^{n^{1-\delta}}$.
\end{thrm}

\begin{proof}
	The idea is to adapt the protocol for Ext-$k$-SAT given in Corollary \ref{coro:ESAT-protocol-final} into a PCP system.	
	
	The message that Bob sends is the same as described in Corollary \ref{coro:ESAT-protocol-final}. 
	The PCP that Alice holds is as follows. For each possible advise string $\mu$ and random string $l$, set $a^{\beta,\mu}_{l}=1$ if Alice accepts under $\mu$ and $l$, otherwise $a^{\alpha,\mu}_{l}=0$. Notice that Bob's message is fixed, and thus it does not affect the length of $a^{\beta,\mu}_{l}$.
	
	The verifier chooses $l\in L$ at random, reads $a^{\beta,\mu}_{l}$, and accepts if and only if $a^{\beta,\mu}_{l}=1$.
	
	The probability that the verifier accepts is exactly equal to the probability that Alice accepts in the protocol described in Corollary \ref{coro:ESAT-protocol-final}.
\end{proof}

\section{Applications}\label{sec:app}

\subsection{Ext-$\rho$-GAP-$k$-SAT}

The first problem to be proved to be sub-linear time inapproximable is Ext-$\rho$-GAP-$k$-SAT, which is analogous to the $\rho$-GAP-$k$-SAT whose hardness directly originates from the classical PCP theorem.

\begin{thrm}\label{thrm:gap-esat-hardness}
	Given a set of assignments $A=\{\alpha_1,\alpha_2,\cdots,\alpha_n \}$, and a $k$-CNF formula $\varphi$ with $d$ variables and $m=O(d)$ clauses where $d=\Omega(\log{n})$ and $d=O(n)$, then assuming SETH, for $\forall \delta>0$, there exists $0<\epsilon <\delta$ and $k\ge 3$ such that no $O(n^{1-\epsilon}d)$ time algorithm could distinguish the following two cases:
	
	(\textbf{Completeness}) there exists an assignment $\alpha^*\in A$ that satisfies $\varphi$;
	
	(\textbf{Soundness}) for $\forall \alpha\in A$, at most $1/2^{n^{1-\delta}}$ portion of clauses in $\varphi$ can be satisfied.
\end{thrm}

\begin{proof}
	According to the proof of Theorem \ref{thrm:sublinear-pcp}, for each possible advice string $\mu$,
	the output of Alice $a^{\beta,\mu}$ is a bit string with length $2^{|l|}$, where $l$ is the random string and $|l|=O(n^{1-\delta}\log{n})$. In analogy with the proof of the hardness of $\rho$-GAP-$k$-SAT, we devise a reduction from the Ext-PCP theorem to Ext-$\rho$-GAP-$k$-SAT. Let $(\varphi, \{\alpha_1,\cdots,\alpha_n \})$ be an instance of Ext-$k$-SAT, and let $N=2^{|\mu|}, D=2^{|l|}$, we construct an instance of Ext-$\rho$-GAP-$k'$-SAT $(\varphi',\{\alpha_1',\cdots, \alpha_N' \})$. 
	Let the $N$ assignments $\alpha'_1,\cdots,\alpha'_N$ be the ones that when $\varphi'$ is assigned with $\alpha'_i$, the satisfiability of each clause exactly corresponds with each bit in $a^{\beta,\mu}$, where $\mu_{(2)}=i_{(10)}$. Fixing a large enough integer $k'\le N$, let $\varphi'$ be a $k'$-CNF with $k'D$ variables and $D$ clauses $\bigwedge\limits_{i=1}^{D}(\bigvee\limits_{j=1}^{k}x_{ij} )$. For $i_{(10)}=\mu_{(2)}$, the $i$-th assignment $\alpha_i'$ is set as follows. If the $j$-th bit of $a^{\beta,\mu}$ is $0$, then set all $x_{ij}=0$, otherwise pick some $j$ and let $x_{ij}=1$. Note that $k$ is large enough to avoid repeated assignments. When all assignments are generated, the redundant clauses and variables can be eliminated. By such construction it can be verified that, if $\varphi$ is satisfied by $\alpha^*\in A$, then $\varphi'$ is satisfiable; if all $\alpha\in A$ can not satisfy $\varphi$, then for each $\alpha'$, at most $1/2^{n^{1-\delta}}$ portion of clauses can be satisfied. 
	
	Then we discuss the parameter range. By $N=2^{|\mu|}, D=2^{|l|}$ and $D=o(N)$ we have $|l|=o(|\mu|)$. Then by $|\mu|=O(n^{1-\epsilon}d)$,  $|l|=O(n^{1-\delta}\log{n})$ and $d=\Omega(\log{n})$ we have $\delta >\epsilon$. 
	
	Finally we analyze the running time. If there exists an $O(N^{1-\epsilon})$ time algorithm $\mathcal{A}$ which can distinguish the YES/NO cases of the above Ext-$\rho$-GAP-$k$-SAT problem, we can substitute the verifier in Theorem \ref{thrm:sublinear-pcp} with $\mathcal{A}$. Since $\mathcal{A}$ runs in $O(N^{1-\epsilon})$ time and $N=2^{|\mu|}$, the algorithm only checks $O(2^{1-\epsilon|\mu|})$ bits in all Alice's outputs. This indicates the length of the advice string can be reduced to $(1-\epsilon)|\mu|$ while still fulfills the protocol for Ext-$k$-SAT. 
	Applying the reduction for multiple times and we can get a new protocol where the length of the advice can be arbitrarily small ($(1-\epsilon)^c|\mu|$, where $c$ can be any integer), which is impossible. 
	
	In conclusion, for $\forall \delta>0$, there exists $0<\epsilon <\delta$ and $k\ge 3$ such that no $O(n^{1-\epsilon}d)$ time algorithm could distinguish the YES/NO cases of the Ext-($1/2^{n^{1-\delta}}$)-GAP-$k$-SAT problem.	
\end{proof}

Next, using the Ext-PCP theorem we will introduce and prove some other sub-linear time inapproximable problems. 

\subsection{Existential Property Test}
Existential Property Test (EPT) is a family of problems given in the following form.

\begin{defn}
	Given a set $S$ of elements over a universe $\mathbb{U}$, a property $P$ defined on $\mathbb{U}$, and $\epsilon>0$, distinguish the following two cases:
	
	(\textbf{Completeness}) there exists an element $e^*\in S$ that satisfies property $P$;
	
	(\textbf{Soundness}) for $\forall e\in S$, $e$ is $\epsilon$-far from satisfying property $P$.
\end{defn}

By the above definition, Ext-$k$-SAT is a special case of EPT where $S$ is a set of binary assignments and $P$ is the $k$-CNF satisfaction property. Next we will describe some other EPT problems, and prove the sub-linear time inapproximability of them.

~\\
\noindent
\textbf{Ext-Max-Inner-Product}
\begin{defn}\label{defn:ext-mip}
	Given a set of bit vectors $V=\{v_1,v_2,\cdots, v_n\}$, and a bit vector $v_0$, find $v^*\in V$ such that the inner product $v^*\cdot v_0$ is maximized.
\end{defn}

\begin{thrm}\label{thrm:ext-mip}
	Let $(V,v_0)$ be an instance of Ext-Max-Inner-Product problem, where $|V|=n$, $|v_i|=d$ and $d=\Omega(\log{n}), d=o(n)$. Assuming SETH and fixing constant $t\le d$, for $\forall \delta>0$, there exists $0<\epsilon<\delta$ such that no $O(n^{1-\epsilon}d)$ time algorithm can distinguish the following two cases:
	
	(\textbf{Completeness}) $\exists v^*\in V$ such that $v^*\cdot v_0\ge t$;
	
	(\textbf{Soundness}) $\forall v\in V$, $v\cdot v_0\le t/2^{n^{1-\delta}}$.
\end{thrm}

\begin{proof}
	Similar with the proof of Theorem \ref{thrm:gap-esat-hardness}, we devise a reduction from Ext-PCP system to Ext-Max-Inner-Product. Let $V$ be the set of Alice's all possible outputs, i.e., $V=\{\mu \mid a^{\beta,\mu} \}$. Let $v_0$ be a vector of $2^{|l|}$ bits 1's, and $t=2^{|l|}$. It can be verified that if there exists $\alpha^*\in A$ satisfying $\varphi$, then there exists $v^*\in V$ such that $v^*\cdot v_0\ge t$. Otherwise, $\forall v\in V, v\cdot v_0\le t/2^{n^{1-\delta}}$. The time complexity argument is similar with proof of Theorem \ref{thrm:gap-esat-hardness} and is omitted.
\end{proof}

There is a pseudo sub-linear time inapproximable result for Ext-Max-Inner-Product in \cite{Abboud2017}. Corollary 1.4 in \cite{Abboud2017} says that assuming SETH and setting $d=N^{o(1)}$, no algorithm can preprocess the problem in polynomial time and distinguish the YES/NO cases in $O(n^{1-\epsilon})$ time with completeness 1 and soundness $1/2^{(\log{n})^{1-o(1)}}$. We have discussed that pseudo sub-linear inapproximable is stronger than pure sub-linear time inapproximable. But this claim only holds when the two results share the same parameters. Actually the soundness in Theorem \ref{thrm:ext-mip}, which is $1/2^{n^{1-\delta}}$ for all $\delta >0$, is higher than Corollary 1.4 in \cite{Abboud2017}, which is $1/2^{(\log{n})^{1-o(1)}}$. Thus we do make some progress on sub-linear time inapproximability by introducing the Ext-PCP theorem.

~\\
\noindent
\textbf{Ext-Max-Vertex-Cover}
\begin{defn}\label{defn:ext-mvc}
	Given a set of vertex sets $\mathbb{V}=\{V_1,V_2,\cdots, V_n \}$, and a graph $G=(V,E)$, find $V^*\in \mathbb{V}$ such that the number of edges in $E$ covered by $V^*$ is maximized.
\end{defn}
\begin{thrm}\label{thrm:ext-mvc}
	Let $(\mathbb{V},G)$ be an instance of Ext-Max-Vertex-Cover, where $|\mathbb{V}|=n$, $|E|=d$ and $d=\Omega(\log{n}),d=o(n)$, $|V_i|=O(d)$. Assuming SETH and fixing $t\le d$, for $\forall \delta > 0$, there exists $0<\epsilon <\delta$ such that no $O(n^{1-\epsilon}d^2)$ time algorithm can distinguish the following two cases:
	
	(\textbf{Completeness}) $\exists V^*\in \mathbb{V}$ such that $V^*$ covers at least $t$ edges in $E$.
	
	(\textbf{Soundness}) $\forall V\in \mathbb{V}$, $V$ covers at most $t/2^{n^{1-\delta}}$ edges in $E$.
\end{thrm}

\begin{proof}
	The proof describes the reduction from Ext-PCP system to Ext-Max-Vertex-Cover. Denoting $N=2^{|\mu|}$ and $D=2^{|l|}$, let $E$ be a set of $D$ edges. For each possible advice $\mu$, let the output of Alice $a^{\beta,\mu}$ be a binary indicator for which edge in $E$ is covered. Then construct the $N$ vertex sets according to each $a^{\beta,\mu}$. The construction may as well be based on the proof of Theorem \ref{thrm:gap-esat-hardness} and follow the classical reduction from SAT to Vertex-Cover. The completeness, soundness and time complexity arguments are similar and thus omitted. The only point need to be explained is that the running time is $O(n^{1-\epsilon}d^2)$, where $d$ is squared. This is because checking if $E$ is covered by $V_i$ needs $O(|E||V_i|)=O(d^2)$ time.
\end{proof}

\subsection{Some Important Observations}
We can get two important observations from the the above two examples. The first is the method to prove sublinear inapproximability, which can be summarized as new reduction. The other is that we have found some linear time solvable but sublinear time inapproximable problems, which is formalized into a new definition.

\subsubsection{Ext-Reduction}
\begin{defn}
	The Ext-Reduction is a special reduction from Ext-$k$-SAT to some EPT problem. The reduction maps the output of Alice in the Ext-PCP theorem into a EPT instance with $N=2^{|\mu|}$ and $D=2^{|l|}$. The instance is constructed so that the if the Ext-$k$-SAT instance is satisfiable, then the EPT instance is satisfiable; if the Ext-$k$-SAT instance is not satisfiable, then the EPT instance is $\epsilon$-far from satisfying the property.
\end{defn}

\subsubsection{Sub-linear time inapproximable problem class}

\begin{defn}(Strictly Linear Time Problems)
	A problem is called in Strictly Linear Time, if there exists $O(nd)$ time algorithm for it,  but there exists $\epsilon >0$ and $\rho<1$ such that no $O(n^{1-\epsilon}d)$ time algorithm can approximate this problem within a factor of $\rho$.  
\end{defn}

By the above definition, Ext-Max-Inner-Product and Nearest-Neighbor problem are known Strictly Linear Time problems.

\begin{defn}(Parameterized Linear Time Problems)
	A problem is called in Parameterized Linear Time, if for fixed $c>1$ there exists $O(nd^c)$ time algorithm for it,  but there exists $\epsilon >0$ and $\rho<1$ such that no $O(n^{1-\epsilon}d^c)$ time algorithm can approximate this problem within a factor of $\rho$.
\end{defn}

By the above definition, Ext-Max-Vertex-Cover is a Parameterized Linear Time problem.

In the above two definitions the parameter $d$ is a measurement of the time to examine one element in the input of size $n$.

\section{Conclusion}\label{sec:conc}
In this paper we proposed the PCP-like theorem for sub-linear time inapproximability. Using the PCP theorem we proved the sub-linear time inapproximability results for several problems. We have shown that the new PCP theorem derives new inapproximability results which has respective advantages compared with existing results. We believe the power of the sub-linear PCP system should be far more than what was revealed in this paper.

\section{Acknowledgment}
Here I sincerely want to thank my tutor Jianzhong Li who is also the communication author of this paper. He has been appealing for the theoretical research on sub-linear time algorithms for years, ever since I first became his doctoral student. There has been a time that I was afraid this work may be too hard for me, but his determination inspired me and drove me on. Finally I am able to finish this work, completing both his and my aspiration on the topic of sub-linear time inapproximability. I must say this work is impossible without him.

We also thank Karthik C. S. and several anonymous reviewers for their valuable advice.

This work was supported by the National Natural Science Foundation of China, grant number 61832003.
\bibliographystyle{plain}
\bibliography{library}

\begin{thebibliography}{10}

\bibitem{Aaronson2009}
Scott Aaronson and Avi Wigderson.
\newblock {Algebrization: A new barrier in complexity theory}.
\newblock {\em ACM Transactions on Computation Theory}, 1(1):1--50, 2009.

\bibitem{Abboud2015}
Amir Abboud, Arturs Backurs, and Virginia~Vassilevska Williams.
\newblock {Tight Hardness Results for LCS and Other Sequence Similarity
  Measures}.
\newblock {\em Proceedings - Annual IEEE Symposium on Foundations of Computer
  Science, FOCS}, 2015-Decem:59--78, 2015.

\bibitem{Abboud2017}
Amir Abboud, Aviad Rubinstein, and Ryan Williams.
\newblock {Distributed PCP Theorems for Hardness of Approximation in P}.
\newblock In {\em 2017 IEEE 58th Annual Symposium on Foundations of Computer
  Science (FOCS)}, volume 2017-Octob, pages 25--36. IEEE, oct 2017.

\bibitem{Agrawal2010}
Parag Agrawal, Arvind Arasu, and Raghav Kaushik.
\newblock {On indexing error-tolerant set containment}.
\newblock In {\em Proceedings of the 2010 international conference on
  Management of data - SIGMOD '10}, page 927, New York, New York, USA, 2010.
  ACM Press.

\bibitem{Alman2015}
Josh Alman and Ryan Williams.
\newblock {Probabilistic Polynomials and Hamming Nearest Neighbors}.
\newblock {\em Proceedings - Annual IEEE Symposium on Foundations of Computer
  Science, FOCS}, 2015-Decem(July):136--150, 2015.

\bibitem{Arora2003}
Sanjeev Arora.
\newblock {How NP got a new definition: a survey of probabilistically checkable
  proofs}.
\newblock III, 2002.

\bibitem{Arora1998a}
Sanjeev Arora, Carsten Lund, Rajeev Motwani, Madhu Sudan, and Mario Szegedy.
\newblock {Proof verification and the hardness of approximation problems}.
\newblock {\em Journal of the ACM}, 45(3):501--555, 1998.

\bibitem{Arora1998}
Sanjeev Arora and Shmuel Safra.
\newblock {Probabilistic checking of proofs: a new characterization of NP}.
\newblock {\em Journal of the ACM}, 45(1):70--122, 1998.

\bibitem{Bern1989}
Marshall Bern and Paul Plassmann.
\newblock {The Steiner problem with edge lengths 1 and 2}.
\newblock {\em Information Processing Letters}, 32(4):171--176, 1989.

\bibitem{Blum1993}
Manuel Blum, Michael Luby, and Ronitt Rubinfeld.
\newblock {Self-testing/correcting with applications to numerical problems}.
\newblock {\em Journal of Computer and System Sciences}, 47(3):549--595, 1993.

\bibitem{Bringmann2015}
Karl Bringmann and Marvin Kunnemann.
\newblock {Quadratic Conditional Lower Bounds for String Problems and Dynamic
  Time Warping}.
\newblock {\em Proceedings - Annual IEEE Symposium on Foundations of Computer
  Science, FOCS}, 2015-Decem:79--97, 2015.

\bibitem{Chazelle2005}
Bernard Chazelle, Ding Liu, and Avner Magen.
\newblock {Sublinear geometric algorithms}.
\newblock {\em SIAM Journal on Computing}, 35(3):627--646, 2005.

\bibitem{Chazelle2005a}
Bernard Chazelle, Ronitt Rubinfeld, and Luca Trevisan.
\newblock {Approximating the Minimum Spanning Tree Weight in Sublinear Time}.
\newblock {\em SIAM Journal on Computing}, 34(6):1370--1379, jan 2005.

\bibitem{Cook1971}
Stephen~A. Cook.
\newblock {The complexity of theorem-proving procedures}.
\newblock In {\em Proceedings of the third annual ACM symposium on Theory of
  computing - STOC '71}, pages 151--158, New York, New York, USA, 1971. ACM
  Press.

\bibitem{Dinur2007}
Irit Dinur.
\newblock {The PCP theorem by gap amplification}.
\newblock {\em Journal of the ACM}, 54(3), 2007.

\bibitem{Ergun1998}
Funda Erg{\"{u}}n, Sampath Kannan, S.~Ravi Kumar, Ronitt Rubinfeld, and Mahesh
  Viswanathan.
\newblock {Spot-checkers}.
\newblock In {\em Proceedings of the thirtieth annual ACM symposium on Theory
  of computing - STOC '98}, pages 259--268, New York, New York, USA, 1998. ACM
  Press.

\bibitem{Fan2013}
Wenfei Fan, Floris Geerts, and Frank Neven.
\newblock {Making queries tractable on big data with preprocessing: (Through
  the eyes of complexity theory)}.
\newblock {\em Proceedings of the VLDB Endowment}, 6(9):685--696, 2013.

\bibitem{Feige2006}
Uriel Feige.
\newblock {On Sums of Independent Random Variables with Unbounded Variance and
  Estimating the Average Degree in a Graph}.
\newblock {\em SIAM Journal on Computing}, 35(4):964--984, jan 2006.

\bibitem{Gao2020}
Xiangyu Gao, Jianzhong Li, Dongjing Miao, and Xianmin Liu.
\newblock {Recognizing the tractability in big data computing}.
\newblock {\em Theoretical Computer Science}, 838:195--207, oct 2020.

\bibitem{Goel2010}
Ashish Goel and Pankaj Gupta.
\newblock {Small subset queries and Bloom filters using ternary associative
  memories, with applications}.
\newblock {\em Performance Evaluation Review}, 38(1 SPEC. ISSUE):143--154,
  2010.

\bibitem{Impagliazzo2001}
Russell Impagliazzo, Ramamohan Paturi, and Francis Zane.
\newblock {Which Problems Have Strongly Exponential Complexity?}
\newblock {\em Journal of Computer and System Sciences}, 63(4):512--530, dec
  2001.

\bibitem{Ivezic2019}
{\v{Z}}eljko Ivezi{\'{c}} and et. al.
\newblock {LSST: From Science Drivers to Reference Design and Anticipated Data
  Products}.
\newblock {\em The Astrophysical Journal}, 873(2):111, mar 2019.

\bibitem{Karp1972}
Richard~M Karp.
\newblock {Reducibility among Combinatorial Problems}.
\newblock In {\em Complexity of Computer Computations}, pages 85--103. Springer
  US, Boston, MA, 1972.

\bibitem{Khachiyan1980}
L.G. Khachiyan.
\newblock {Polynomial algorithms in linear programming}.
\newblock {\em USSR Computational Mathematics and Mathematical Physics},
  20(1):53--72, jan 1980.

\bibitem{Levin1973}
L~A Levin.
\newblock {Universal Sequential Search Problems}.
\newblock In {\em Probl. Peredachi Inf.}, volume~9, pages 115--116, 1973.

\bibitem{Li1987}
Jian-zhong Li and Harry~K.T. Wong.
\newblock {Batched Interpolation Searching on databases}.
\newblock In {\em 1987 IEEE Third International Conference on Data
  Engineering}, number July, pages 18--24. IEEE, feb 1987.

\bibitem{Melnik2003}
Sergey Melnik and Hector Garcia-Molina.
\newblock {Adaptive algorithms for set containment joins}.
\newblock {\em ACM Transactions on Database Systems}, 28(1):56--99, mar 2003.

\bibitem{SortBenchmark}
Chris Nyberg and Mehul Shah.
\newblock {Sort Benchmark Home Page}, 2019.

\bibitem{Papadimitriou1988}
Christos~H. Papadimitriou and Mihalis Yannakakis.
\newblock {Optimization, approximation, and complexity classes}.
\newblock {\em Proceedings of the Annual ACM Symposium on Theory of Computing},
  pages 229--234, 1988.

\bibitem{Papadimitriou1993}
Christos~H Papadimitriou and Mihalis Yannakakis.
\newblock {The Traveling Salesman Problem with Distances One and Two}.
\newblock {\em Math. Oper. Res.}, 18(1):1--11, 1993.

\bibitem{Parnas1999}
Michal Parnas and Dana Ron.
\newblock {Testing the diameter of graphs}.
\newblock {\em Lecture Notes in Computer Science (including subseries Lecture
  Notes in Artificial Intelligence and Lecture Notes in Bioinformatics)},
  1671:85--96, 1999.

\bibitem{Ramasamy2000}
Karthikeyan Ramasamy, Jignesh~M. Patel, Jeffrey~F. Naughton, and Raghav
  Kaushik.
\newblock {Set containment joins: The good, the bad and the ugly}.
\newblock {\em Proceedings of the 26th International Conference on Very Large
  Data Bases, VLDB'00}, pages 351--362, 2000.

\bibitem{Roditty2013}
Liam Roditty and Virginia~Vassilevska Williams.
\newblock {Fast approximation algorithms for the diameter and radius of sparse
  graphs}.
\newblock {\em Proceedings of the Annual ACM Symposium on Theory of Computing},
  pages 515--524, 2013.

\bibitem{Rubinfeld2011}
Ronitt Rubinfeld and Asaf Shapira.
\newblock {Sublinear time algorithms}.
\newblock {\em SIAM Journal on Discrete Mathematics}, 25(4):1562--1588, 2011.

\bibitem{Karthik2019}
Karthik~C. S., Bundit Laekhanukit, and Pasin Manurangsi.
\newblock {On the Parameterized Complexity of Approximating Dominating Set}.
\newblock {\em Journal of the ACM}, 66(5):1--38, sep 2019.

\bibitem{Williams2018}
Virginia~Vassilevska Williams.
\newblock {On some fine-grained questions in algorithms and complexity}.
\newblock {\em Proceedings of the International Congress of Mathematicians, ICM
  2018}, 4:3465--3506, 2018.

\end{thebibliography}
\end{document}